\documentclass[10pt, english,prl,aps,twocolumn,floatfix,superscriptaddress]{revtex4-1}
\usepackage[T1]{fontenc}
\usepackage[latin9]{inputenc}
\usepackage{color}
\usepackage{babel}
\usepackage{verbatim}
\usepackage{refstyle}
\usepackage{amsthm}
\usepackage{amsmath}
\usepackage{graphicx}
\usepackage{color}
\usepackage[pdfstartview={FitH},bookmarks=false]
 {hyperref}

\theoremstyle{plain}
\newtheorem{thm}{\protect\theoremname}
  \theoremstyle{remark}

  \providecommand{\remarkname}{Remark}
\providecommand{\theoremname}{Theorem}

\newcommand{\Tes}{\Upsilon}

\makeatletter

\newcommand{\ket}[1]{\left| #1\right\rangle}        
\newcommand{\bra}[1]{\left\langle #1\right|}        
\begin{document}

\title{Novel Technique for Robust Optimal Algorithmic Cooling}

\author{Sadegh Raeisi}
\email[]{sadegh.raeisi@gmail.com}
\affiliation{Department of Physics, Sharif University of Technology, Tehran, Iran}
\author{M\'{a}ria Kieferov\'{a}}
\affiliation{Institute for Quantum Computing, University of Waterloo, Ontario N2L 3G1, Canada}
\affiliation{Department of Physics and Astronomy, University of Waterloo, Ontario N2L 3G1, Canada}
\affiliation{Department of Physics and Astronomy, Macquarie University, Sydney New South Wales 2109, Australia}
\author{Michele Mosca}
\affiliation{Institute for Quantum Computing, University of Waterloo, Ontario N2L 3G1, Canada}
\affiliation{Department of Combinatorics and Optimization, University of Waterloo, Ontario N2L 3G1, Canada}
\affiliation{Perimeter Institute for Theoretical Physics, Waterloo, Ontario N2L 2Y5, Canada}
\affiliation{Canadian Institute for Advanced Research, Toronto, Ontario M5G 1Z8, Canada}

\begin{abstract}
Heat-bath algorithmic cooling (HBAC)  
provides algorithmic ways to improve the purity of quantum states.
These techniques are complex iterative processes that 
change from each iteration to the next and this poses a 
significant challenge to implementing these algorithms. 
Here, we introduce a new technique that on a fundamental 
level, shows that it is possible to do algorithmic cooling and 
even reach the cooling limit without any knowledge of the state 
and using  only a single fixed operation,  
and on  a practical level, presents a more feasible and 
robust alternative for implementing HBAC. 
We also show that our new technique converges to the asymptotic state
of HBAC and that the cooling algorithm can be efficiently implemented; 
however, the saturation could require exponentially many 
iterations and remains impractical.
This brings HBAC to the realm of feasibility and 
makes it a viable option for realistic application 
in quantum technologies. 
\end{abstract}
\maketitle

Many quantum effects and quantum technologies rely on fragile quantum 
fluctuations that can easily be overwhelmed by thermal fluctuations. 
This is why often 
techniques for suppressing thermal fluctuations such as 
cooling in a cryostat or a dilution fridge are required. 
There are, however, dynamical cooling techniques 
that more surgically extract energy from subsystems 
of interest and can lower the temperature 
beyond what would be feasible with conventional 
cooling of the entire system.

Heat-bath algorithmic cooling (HBAC) are 
techniques that operate on an ensemble 
of qubits and effectively cool down and purify a target 
subset of  qubits. 
HBAC drives the system out of equilibrium by transferring 
the entropy from the target qubits to the rest of the ensemble. 
This is often referred to as ``compression'' 
since it uses information
theoretical techniques to 
compress the entropy to the nontarget elements of the ensemble 
and effectively cools down the target qubits. 
The target and the refrigerant qubits are often 
referred to as the ``computation''  
and the ``reset'' qubits respectively. 

HBAC can be seen as an extension of techniques like  
DNP or INEPT \cite{kuhn2013hyperpolarization} to situations 
where there is access to more than 
two spin species and thus could in principle
go beyond the purity of the reset qubit. While
applications of HBAC go beyond a specific implementation like NMR, 
it could be combined with techniques known in each implementation; e.g. 
DNP can be used to provide the source polarization of HBAC in NMR. 

Algorithmic cooling was first introduced
in \cite{schulman_scalable_1998} for a closed system, 
for which, the cooling 
is limited by the Shannon bound for information compression.
This process heats up the reset qubits beyond 
their initial temperature. 
It was later  proposed to use a heat-bath to recycle the reset 
qubits and enhance the cooling beyond the Shannon 
bound \cite{boykin_algorithmic_2002}. 
In this setting, 
the reset qubits, through the interaction with a 
heat-bath,  are  cooled down to the heat-bath temperature again. 
This is known as the ``reset step''.  

Similar settings has been investigated in the context of 
Quantum Thermodynamics(QT) and dynamic cooling 
\cite{janzing2000thermodynamic, scharlau2016quantum, boes2018neumann,
streltsov2018maximal, masanes2017general, allahverdyan2011thermodynamic}. 
The cooling limit, the corresponding resource 
theories and generalizations of the third law of thermodynamic
are among topics that have attracted a lot of attention
in QT 
\cite{ scharlau2016quantum, reeb2014improved,browne2014guaranteed,
allahverdyan2011thermodynamic, masanes2017general, streltsov2018maximal}.
It turns out that some of these results, like the existence of 
a cooling limit, apply to HBAC too. 
Interestingly, even with the help of a heat-bath, 
it is not possible to extract
all the entropy from the target qubits  
 \cite{schulman_physical_2005}. 
The optimal technique was introduced by Schulman \textit{et al.} 
in \cite{schulman_physical_2005} 
and is known as the partner pairing algorithm (PPA).
The existence of the limit was 
proved by Schulman \textit{et al.} 
in \cite{schulman_physical_2005}.
Later Raeisi and Mosca 
\cite{raeisi2015asymptotic} 
established the asymptotic limit of PPA,  
with the corresponding asymptotic state, 
and proved that the process asymptotically 
approaches the cooling limit.
We refer to the optimal asymptotic cooling state as OAS. 

One of the main challenges with HBAC techniques, 
especially the ones that converge to OAS, is that 
they are highly complex. 
The operations change from each iteration to the next 
and in many cases, there is no recipe for 
implementing the operations in each step. 

For instance, PPA sorts the diagonal elements of 
the density matrix in each iteration. But this means 
not only that one needs to know the state in each iteration, 
but also that the operation for implementing the sort would 
change in each iteration as the state changes. 

An obvious question is whether or not it would be possible 
to reach the OAS and the cooling limit with a fixed 
state-independent operation in each iteration. 
Note that the state is constantly changing through the
cooling process, and 
naturally, the compression should change too, 
as is the case with PPA. 
Reaching the OAS with a fixed operation seems even more non-trivial. 
In other words, for a HBAC technique with fixed operation, 
the compression should be tuned such that 
without knowing the state, 
not only can it extract entropy from the state, 
but  through the repetition of the process, 
it would push the state into the OAS. 
In the language of QT, this translates to finding a 
periodic evolution that would 
make an optimal cyclic cooling process. 
The typical complex non-periodic evolutions of the HBAC process
make it challenging to draw direct connections between QT and HBAC.
Existence of a cyclic HBAC 
technique with fixed iteration would go a long way in bridging 
this gap.

Here we answer this question and show that this is in fact possible. 
We introduce a compression operation that can push the state to 
the OAS and reach the cooling limit of HBAC and makes 
a state-independent cycle process. 
Further, we show that it can be implemented efficiently 
and give a recipe for building the quantum circuit. 

Besides the fundamental significance, this result  
could have a critical impact on the feasibility of HBAC 
techniques. First, in contrast to techniques such as 
PPA, our algorithm can be efficiently implemented. 
We, however, show that reaching the OAS would require 
exponentially many iterations. Second, 
the state-independence of operations 
makes our algorithm simple and more robust and turns HBAC 
to a viable option
for generating large scale supplies of high-purity
quantum states.

We start by introducing our algorithm and then present the 
complexity analysis. Next, we compare it against PPA. 
We then investigate the robustness of the two techniques. 

 
We assume an ensemble of $n+1$ qubits, with the first $n$
as the computation and the last as the reset qubits. 
We use the subscript $R$ and $C$ to refer to the 
reset and the computation qubits. 
We also assume that the Hilbert space is ordered as 
$\mathcal{H}_{C}\otimes\mathcal{H}_{R}$;  
the first part is the computation qubits and the last part 
are  the reset qubits.


In our technique, instead of sorting the diagonal elements, we apply the 
following unitary for compression in each iteration: 

\begin{equation}
U_{\text{TS}}=\left(\begin{array}{ccccc}
1\\
 & X\\
 &  & \ddots\\
 &  &  & X\\
 &  &  &  & 1
\end{array}\right),\label{eq:Bi-sort_U}
\end{equation}
where $X$ is the Pauli $X$ operator and the first and the last 
elements of the matrix are one. 
The matrix is $2^{n+1}\times2^{n+1}$
and acts on both the computation and the reset qubits. 
We refer to the unitary $U_{\text{TS}}$ as the two-sort operator
and to our technique as  two-sort algorithmic cooling (TSAC). 
The unitary $U_{\text{TS}}$ swaps every two neighboring 
elements on the diagonal of the density matrix, 
except for the first and the last elements. 
Intuitively, this is a partial sort that acts 
locally on the density matrix. 
This is the golden operation that makes it possible 
to reach the cooling limit without knowing the state. 

After compression,  the reset qubit is reset which is equivalent to 
 $\mathcal{R}\left[\rho\right]=\text{Tr}_{R}\left(\rho\right)\otimes\rho_{R}$,  where
$\text{Tr}_{R}$ is the partial trace over the 
reset qubit and $\rho_{R}$ is 
the ``reset state'' 
\begin{equation}
\rho_{R}=\frac{1}{z}\left(\begin{array}{cc}
e^{\epsilon} & 0\\
0 & e^{-\epsilon}
\end{array}\right),\label{eq:reset-state}
\end{equation}
with $z=\left(e^{\epsilon}+e^{-\epsilon}\right)$. 
The parameter $\epsilon$ is called the polarization. 
Our method does not make any non-trivial assumption 
about $\epsilon$ nor $n$.

Mathematically,  each iteration applies the following 
channel on the full density matrix
\begin{equation}
\label{eq:TSAC_Iteration}
\mathcal{C}_{\text{TS}}\left[\rho\right]=U_{\text{TS}}^{\dagger}\left[\text{Tr}_{R}\left(\rho\right)\otimes\rho_{R}\right] U_{\text{TS}}.
\end{equation}

This process is clearly independent of the iteration or the state
 and it can be described by a time-homogeneous Markov process. 
We find the transfer matrix  and  use its 
spectrum to prove that the process
converges to the OAS and to give an upper-bound for the required number of iterations.

The sequences of the elements on the diagonal of the 
density matrix form a Markov chain. 
We use  a vector with $2^{n+1}$ elements $\left\{ \lambda^t\right\}$ to represent  
the state after the $t$th iteration. 
We use a similar notation for the density matrix of 
the computation qubits (without the
reset qubit) and use $\left\{ p^t\right\} $ 
instead.

Figure \ref{fig:Pictorial_Iterations} gives a pictorial 
description of the process in each iteration. 
It starts with the sequence $\{ \lambda^{t} \}$,  
the diagonal elements of the density matrix of the $n$ 
computation and one reset qubit in the $t{\text{th}}$ iteration. 
First,  there is the reset step which 
takes the reset qubit to the 
state in Eq. (\ref{eq:reset-state}). 
This takes every two neighbouring elements $\lambda^{t}_{2k+1}$ 
and $\lambda^{t}_{2k+2}$ to $p^t_k= \lambda^{t}_{2k+1}+\lambda^{t}_{2k+2}$ 
and then splits them into $\zeta^{t}_{2k+1}=p^t_k e^{ \epsilon}/z$ 
and $\zeta^{t}_{2k+2}=p^t_k e^{- \epsilon}/z$. 
Now the two-sort unitary is applied and rearranges 
the array to $\{ \lambda^{t+1} \}$ such that  
$\lambda^{t+1}_{2k}=\zeta^{t}_{2k+1}$ and $\lambda^{t+1}_{2k+1}=\zeta^{t}_{2k}$.

\begin{figure}
\begin{centering}
\includegraphics[width=\columnwidth]{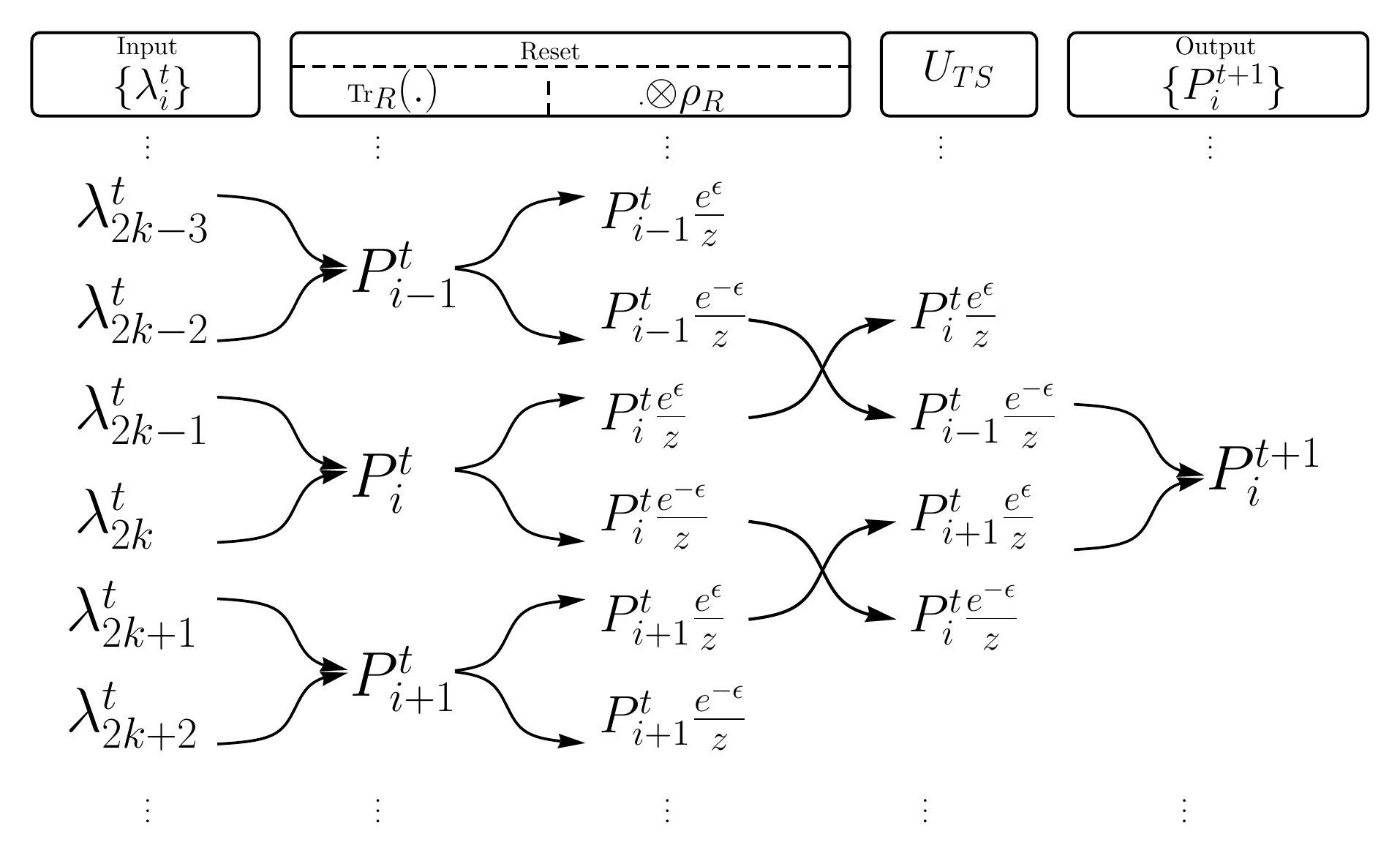}
\par\end{centering}

\caption{\label{fig:Pictorial_Iterations} The pictorial description of an iteration. 
The input is the list of the diagonal elements of the full density matrix, 
$\{ \lambda^{t} \}$. The reset step first merges every 
two neighbouring elements (partial trace) 
and then  replaces the reset qubit with $\rho_R$ which
splits each element into two elements again. 
Next, the $U_{\text{TS}}$ swaps all the neighbouring elements
except for the first and the last one.
}
\end{figure}

For simplicity, we focus on the computation qubits and trace out the reset qubit. This gives the following update rule for the diagonal elements of the computation qubits:
\begin{equation}
p_{i}^{t+1}=p_{i-1}^{t}\frac{e^{-\epsilon}}{z}+p_{i+1}^{t}\frac{e^{\epsilon}}{z}.\label{eq:Update_Rule}
\end{equation}
for $1<i<2^n$.

Similarly, for the first and the last elements, the update rules are 
$p_{1}^{t+1}=(p_{1}^{t}+p_{2}^{t})\frac{e^{\epsilon}}{z}$ and
$p_{2^{n}}^{t+1}=(p_{2^{n}-1}^{t}+p_{2^{n}}^{t})\frac{e^{-\epsilon}}{z}$.

These update rules give the following $2^n\times 2^n$ transition matrix for the Markov process:
\begin{equation}
T=\frac{1}{z}\left(\begin{array}{ccccc}
e^{\epsilon} & e^{\epsilon} & 0 & \cdots & 0\\
e^{-\epsilon} & 0 & e^{\epsilon} & \cdots & 0\\
0 & e^{-\epsilon} & 0 & \cdots & 0\\
0 & 0 & \cdots & \ddots & \vdots\\
0 & 0 & \cdots & e^{-\epsilon} & e^{-\epsilon}
\end{array}\right).\label{eq:MarkovMatrix}
\end{equation} 

It is easy to verify that $\vec{\{p^{t+1}\}} =T \vec{\{p^t \}}$ 
and gives the update rules above. 
The matrix $T$ has a unique eigenvalue $1$ and the remaining eigenvalues are  
 $\Tes_{k}=\frac{2\cos{\frac{k\pi}{2^n}}}{z}$ for $k=1,2, \cdots, 2^n-1$. 
 The eigenstate corresponding to  eigenvalue  one is
\begin{equation}
\rho =p_0 \lbrace1,e^{-2\epsilon}, e^{-4\epsilon}, \cdots\rbrace,
\label{OAS}
\end{equation}
which is the OAS and $p_0$ is the normalization 
factor \cite{raeisi2015asymptotic}.  
For the detailed calculation of the eigensystem, see the Supplemental Material(SM). 
Since all the other eigenvalues 
lie in the interval $(1, -1)$, the Markov chain asymptotically converges to
$\rho$. This proves that our technique asymptotically achieves the cooling 
limit of HBAC.

We give a $O\left(n^{2}\right)$  
circuit for the implementation of the two-sort unitary in the SM. 
We first shift the basis 
by one, which transforms it to $\text{Toff}(n+1) \sigma_x^{(n+1)}$. 
Then a Pauli $\sigma_x$ on the reset qubit 
turns the matrix to a multiple-control Toffoli gate (see SM for details). 
This shows that our technique can be efficiently implemented. 
However, to reach the OAS, we need to investigate how many 
iterations would be required. 
The mixing time of a Markov chain is the number of iterations required to get within
distance $\xi$ of the asymptotic state (i.e. to achieve $| \rho^{t} - \rho_{OAS} | \leq \xi$).
We can upper bound this number of iterations as a function of the spectral gap $\Delta$, i.e. 
the difference between $1$ and the second largest eigenvalue,
\begin{equation}
\label{eq:mixingtime_upperbound}
t_{\text{mix}}\left(\xi\right)\leq\log\left(\frac{1}{\xi\, l}\right)\frac{1}{\Delta},
\end{equation}
where 
$l=p_0 e^{-(2^n-1)\epsilon}$ 
 is the smallest element of the array in Eq. (\ref{OAS})  \cite{levin2009markov}.

The spectral gap is $\Delta = 1-\left(2\cos{\frac{\pi}{2^n}}\right)/\left(e^{\epsilon}+e^{-\epsilon}\right)\geq \frac{z-2}{z}$.
This gives 
\begin{equation}
t_{\text{mix}} (\xi) \leq \left( \log\left(\frac{1}{\xi \, p_0 e^{-(2^n-1)\epsilon}}\right) \left(\frac{z}{z-2}\right) \right).
\end{equation}
It is easy to check that  $t_{\text{mix}} \in O(2^{n})$ 
which yields 
 $O(n^2 2^n)$ for the overall complexity  of TSAC. 

Now, we compare our technique with PPA. 
The key ingredient of PPA is 
  sorting the diagonal elements of 
the density matrix in each iteration. 
This transfers as
much  entropy as possible from the computation elements 
to the reset qubit \cite{schulman_physical_2005}.  
Then the reset qubit is reset back to its equilibrium state.

Even assuming that finding the operation for 
sorting a $2^{n+1}$ array is easy,
we need to find 
quantum circuits to implement them which has at least
$O(e^n)$ classical complexity. 
Note that this is only the classical cost of the 
algorithm. Without this, it would not be even 
possible to start implementing PPA. 
For systems as large as 20-40 qubits, 
e.g. the experiment in \cite{pande2017strong},   
not only it is challenging to implement PPA,  
but it also seems difficult to find the required permutations.
This is in contrast to our technique, where each 
iteration is already known and there is a specific 
circuit for implementing it.


Next is the gate complexity of PPA. 
Typically only the number of iterations is counted, 
ignoring the complexity of the sort operations. 
In fact, due to the complexity of
the sort operations, it is difficult to  bound the 
number of gates required for PPA. 
Naively, there are $(2^n)!$ permutation matrices of size $2^n$. 
Assuming a finite number of one- and two-qubit gates, 
there are only $\left[\text{poly}(n)\right]^d$ circuits with at 
most $d$ gates on $n$ qubits. Taking $d$ to be $\text{poly}(n)$, 
one can see that only a small fraction of permutations 
can be implemented efficiently.
Here, we provide a more rigorous bound on 
the gate complexity of PPA.

The permutation operations can be decomposed into separate 
cycles that form disjoint blocks in the permutation matrix. 
These cycles could have different sizes and the size of each cycle 
determines the number of states it permutes cyclically. 
These are known as $k$-cycles, 
where $k$ is the size of the block.

Assume that for all the $k$-cycles in the permutation, 
we can find an efficient circuit. 
Also assume that, given a certain state $\ket{j}$, it is possible to 
efficiently determine which block the state belongs to. 
Furthermore, we assume that cycles of equal size can be 
implemented in parallel efficiently. 
These assumptions may not be true, 
but any lower bound established with these assumptions 
still holds when any of these assumptions are weakened or dropped.
Under these assumptions, the cost would depend on the 
number of $k$-cycles with distinct $k$ values. 
This is the number of blocks in the permutation matrix 
that have different size. 
We refer to this quantity as NBDS. 

Implementation of each sort operation requires 
the implementation of all the blocks. Blocks of 
different size cannot be fully parallelized and 
for switching between each of two blocks of unequal size, 
some quantum operation would be required. 
This sets the number of blocks of 
different size NBDS as a lower bound for the complexity of 
any sort operation.

Figure (\ref{fig:NBDS}) shows the simulation 
results of PPA for a different number 
of computation qubits, $n$  and indicates that 
NBDS grows exponentially with $n$. This implies that our lower-bound for 
the gate complexity of PPA scales exponentially with $n$.  
Here, for any value of $n$, we get a sequence of permutation matrices and 
  pick the permutation that has the largest NBDS. 
\begin{figure}
\begin{centering}
\includegraphics[trim={0cm 0cm 0cm 0cm},clip ,width=\columnwidth]{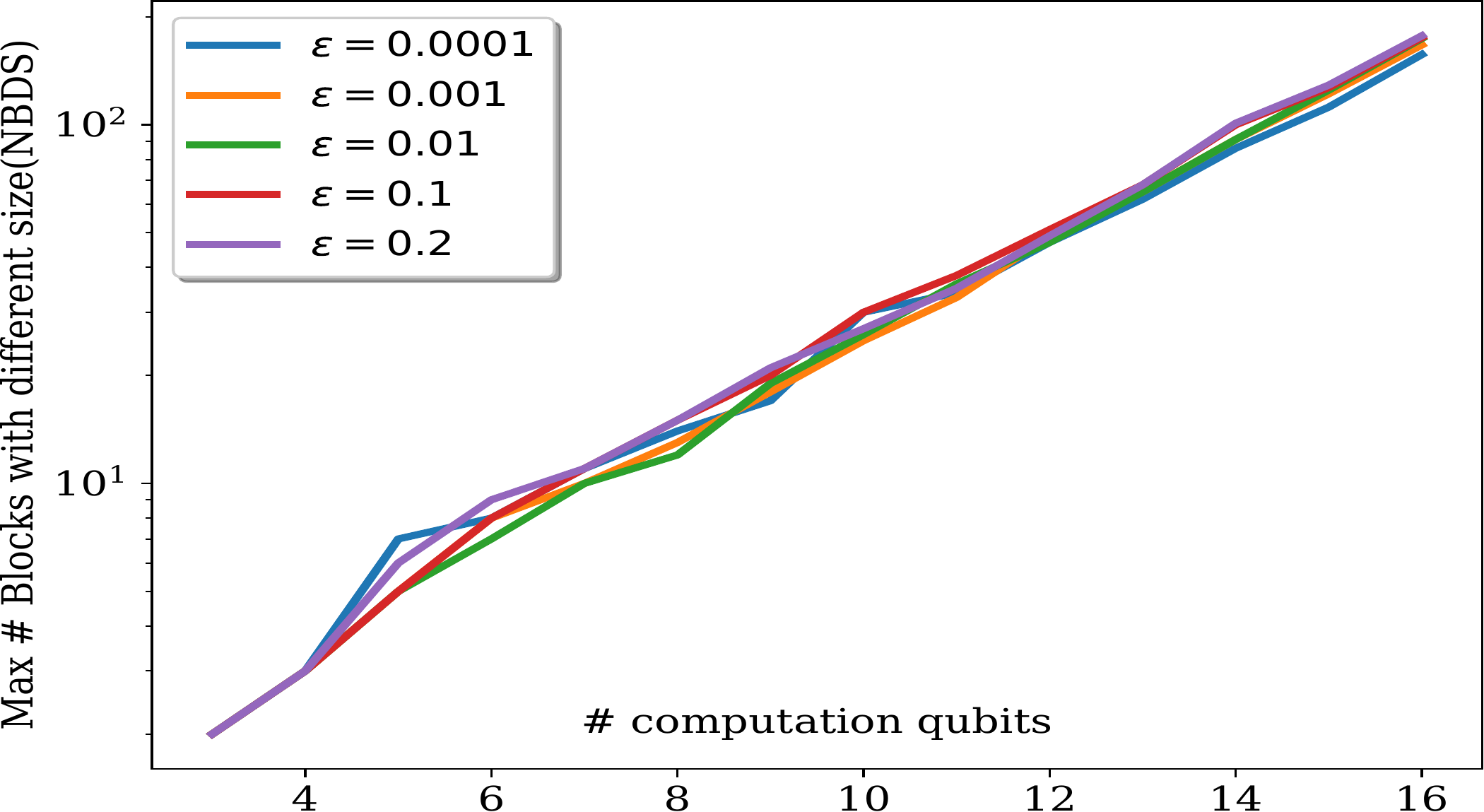}
\par\end{centering}

\caption{\label{fig:NBDS} This plot shows that the maximum of 
the NBDS for implementations of PPA grows exponentially with 
the number of computation qubits $n$.  The $y$ axis a 
is logarithmic scale.  Different plots correspond to different
reset polarizations $\epsilon$. 
 }
\end{figure}

Last, there is the fragility  to practical  
imperfections. 
The sort operation of PPA requires the ordering of 
the diagonal elements of the state. This means that 
techniques like quantum state tomography are needed 
to monitor the state. This process however, cannot 
be perfect and usually there are estimation errors. 
Figure (\ref{fig:Noisy_Tomography})  shows 
how sensitive the process is to these imperfections. 
These simulations are for HBAC with $n=2$ and one reset 
qubit and the reset polarization of $\epsilon = 0.02$ with 
a zero-mean Gaussian noise with variance $\sigma$. 
As $\sigma$ increases, the process becomes random and 
would not approach the cooling limit any more.  
Figure (\ref{fig:Noisy_Tomography}) also shows the result for TSAC 
which regardless of 
the noise, would always converge to the OAS. 
This is not a generic noise model, but is  
relevant for techniques like PPA and shows that with noise,
PPA can heat the state.
Next, we show that TSAC  monotonically pushes the state towards the OAS. 

\begin{thm}\label{thm:Robustness}
Given some state $\rho$ and a reset state $\rho_R$ with 
polarization $\epsilon$, 
if the polarization of the first 
computation qubit is less than the HBAC limit, each iteration 
of TSAC, as in equation (\ref{eq:TSAC_Iteration}) would increase
the polarization of the first qubit. 
\end{thm}
\begin{proof}
The polarization of the first qubit is determined by the first half of 
the diagonal elements of $\rho$ and we need to show that TSAC would 
increase it. 
The assumption that the state is hotter than the HBAC limit means
\begin{equation}
p_{2^{n-1}}e^{-\epsilon} < p_{2^{n-1}+1}e^{\epsilon}.
\end{equation}
After the iteration, these two are swapped by the $U_{\text{TS}}$. This
increases the sum of the first half of the diagonal elements and 
as a result, the polarization of the first qubit. 
\end{proof}
Note that this can be extended to other computation qubits. 
It is just easier to show for the first qubit. 

\begin{figure}
\begin{centering}
\includegraphics[trim={0cm 0cm 0cm 0cm},width=\columnwidth]{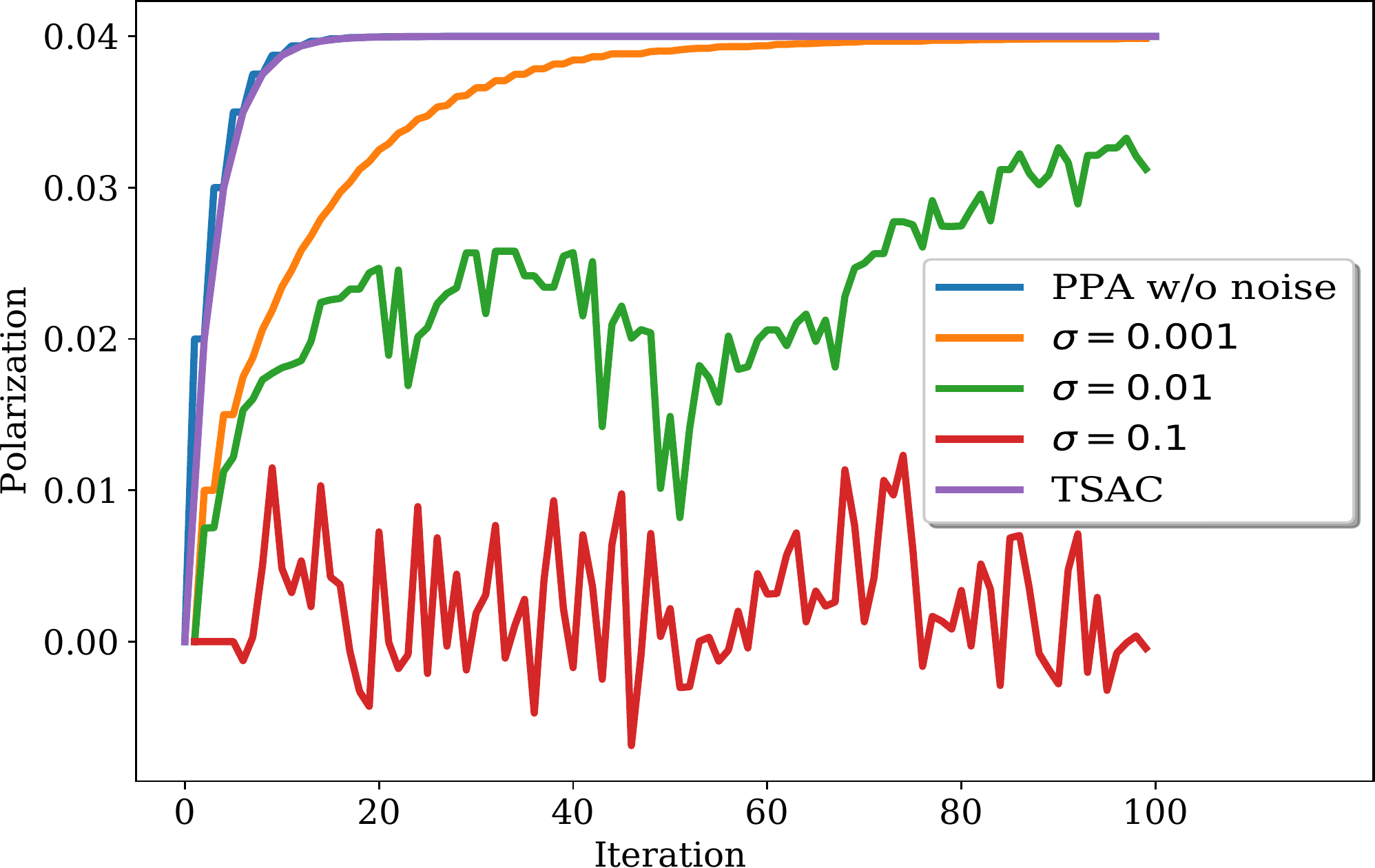}
\par\end{centering}

\caption{\label{fig:Noisy_Tomography} 
This plot shows the polarization of the first qubit vs iterations. 
This is for the simulation of PPA for two computation qubits and one
reset qubit, with different amounts of the state estimation errors $\sigma$.  }
\end{figure}

Since the process does not depend on the state, 
even when the state is perturbed from the ideal one, 
the process continues to cool it down. 
Note that this does not imply that our algorithm 
is robust to all imperfections. Specifically, 
with faulty operations, no algorithm can 
guarantee the convergence to OAS.

Table (\ref{tab:Comparison}) gives a comparison between TSAC and PPA. 
Here, the ``noise sensitivity'' is the sensitivity to deviations
from the expected state in the process. 
The table  demonstrates that
TSAC outperforms PPA in almost every aspect 
and presents a more realistic option in practice.

\begin{table}
\tiny
\begin{centering}
\resizebox{.9\columnwidth}{!}{%
\begin{tabular}{|c|c|c|}
\hline 
 & PPA & TSAC\tabularnewline
\hline 
\hline 
\shortstack{Classical  cost \\  of circuit synthesis
}  & $\Omega(2^n)$ & $\text{poly}(n) $\tabularnewline
\hline 
\shortstack{Total number \\ of gates} & \shortstack{(Conjecture)\\ $\Omega(2^{n})$} & $O\left(n^2 2^{n}\right)$\tabularnewline
\hline 
Noise sensitivity & Sensitive& Robust \tabularnewline
\hline 
\end{tabular}%
}
\par\end{centering}
\caption{\label{tab:Comparison}Comparison between PPA and our technique. 
Here, the ``noise sensitivity" is the sensitivity to deviations
from the expected state in the process. }
\end{table}

In conclusion, our work presents a novel viable technique for optimal 
HBAC which shows that optimal HBAC is possible without any 
knowledge of the state and without changing the operation 
through the process. 
From a Quantum Thermodynamics(QT) viewpoint, 
it means that the optimal cooling is possible to do HBAC 
 with a cyclic process. 
This opens new avenues for examining HBAC in terms of QT.
For instance, resources 
required for reaching the cooling limit have been 
extensively investigated in QT 
\cite{ scharlau2016quantum, 
allahverdyan2011thermodynamic, masanes2017general, streltsov2018maximal,
 reeb2014improved, browne2014guaranteed,masanes2017general}. 
It is interesting to map the required resources 
and their scaling to HBAC.

Our work also brings 
realistic applications of these techniques
to the realm of possibility. 
The new technique, in contrast to PPA,  uses a fixed operation in every iteration, 
which addresses the fragility issues in previous works. 
More precisely, the new technique is robust 
against imperfections and noise in the state.

It is also possible to combine our work with other dynamics cooling 
techniques 
to further reduce the costs. 
However, it remains open to see
how far the complexity may be reduced.
Results from QT on the analysis of the 
resources for cooling and the extensions of the third 
law of thermodynamics
\cite{ masanes2017general, scharlau2016quantum, 
allahverdyan2011thermodynamic,streltsov2018maximal}
 could prove helpful for reducing the complexity of cost.

\begin{acknowledgments}
We thank Alex Parent for helpful discussions. 
This work was supported by the research 
grant system of Sharif University of Technology (G960219),
ERC Starting Grant OPTOMECH, 
Canada's NSERC, CIFAR and CFI.
IQC and Perimeter Institute are supported in part by the Government
of Canada and the Province of Ontario. 
\end{acknowledgments}

\bibliographystyle{aipnum4-1}
\bibliography{Cooling_Bound_AC}

\onecolumngrid
\vspace{20cm}

\setcounter{equation}{0}
\setcounter{figure}{0}
\setcounter{table}{0}
\makeatletter
\renewcommand{\theequation}{S\arabic{equation}}
\renewcommand{\thefigure}{S\arabic{figure}}
\renewcommand{\bibnumfmt}[1]{[S#1]}
\renewcommand{\citenumfont}[1]{S#1}

\section{Supplementary Material}

\subsection*{Spectrum of the transfer matrix}

We solve the eigenvalue equation, $T \Phi^{(p)}=\Tes \Phi^{(p)}$, 
indexing the eigenvectors by $p$. 
Using the sparsity and the structure of $T$, we can rewrite 
the eigenvalue equations as

\begin{align}
\Tes^{(p)}\Phi_{1}^{(p)} &=(\Phi_{1}^{(p)}+\Phi_{2}^{(p)})\frac{e^{\epsilon}}{z} \label{eq:FirstT},\\
\Tes^{(p)}\Phi_{k}^{(p)} &= \Phi_{k-1}^{(p)}\frac{e^{-\epsilon}}{z}+\Phi_{k+1}^{(p)}\frac{e^{\epsilon}}{z}\label{eq:Update_RuleT},\\
\Tes^{(p)}\Phi_{2^{n}}^{(p)} &=(\Phi_{2^{n}-1}^{(p)}+\Phi_{2^{n}}^{(p)})\frac{e^{-\epsilon}}{z}\label{eq:LastT},
\end{align}
for $1<k<2^n$.
We use the ansatz 
\begin{equation}
\Phi_k^{(p)} = e^{(ip-\epsilon)k} + \alpha e^{(-ip-\epsilon)k}
\end{equation}
with arbitrary complex parameters  $\alpha$ and $p$. 
This ansatz automatically satisfies \ref{eq:Update_RuleT} with eigenvalue
\begin{equation}
\Tes^{(p)}=\frac{2\cos{p}}{e^{\epsilon}+e^{-\epsilon}}.
\end{equation}

We set the value of $\alpha$ by solving \ref{eq:FirstT} and obtain $\alpha=\frac{e^{ip}-e^{-\Delta}}{e^{-\Delta}-e^{-ip}}$.
Note that this result forbids $p=0$ because it gives $\Phi_k^{(0)}=0$. 

At last, we satisfy \ref{eq:LastT} and obtain allowed values of $p$. 
The solution $ip=\pm\epsilon$ gives eigenvalue $1$ and corresponds to the eigenvector  $\Phi_{k}^{(p)} = e^{-2\Delta k}$. 
This is the asymptotic state of PPA as was proved in \cite{raeisi2015asymptotic}. 

The remaining eigenvalues are of form $\frac{2\cos{\frac{j\pi}{2^n}}}{e^{\epsilon}+e^{-\epsilon}}$ for $1\leq j < 2^n$. 
All these eigenvalues lie in the range $(-1,1)$. In other words, the Markov chain has a unique eigenvalue one and all the other eigenvalues are smaller than one. Therefore the Markov chain defined by the transition matrix $T$ converges to the +1 eigenvector, which is OAS.

The convergence rate is determined by the difference between $1$ and the second largest eigenvalue, 
  $\Tes^{(2)}=\frac{2\cos{\frac{\pi}{2^n}}}{e^{\epsilon}+e^{-\epsilon}}$. We can bound the gap as 
  $\Delta = 1-\left(2\cos{\frac{\pi}{2^n}}\right)/\left(e^{\epsilon}+e^{-\epsilon}\right)\geq \frac{z-2}{z}$.
The mixing time is then upper-bounded by 
\begin{equation}
\label{eq:SM_mixingtime_upperbound}
t_{\text{mix}}\left(\xi\right)\leq\log\left(\frac{1}{\xi\, l}\right)\frac{1}{\Delta} \leq \left( \log\left(\frac{1}{\xi\, l}\right) \left(\frac{z}{z-2}\right) \right) \leq c_1 \log\left(\frac{1}{ l}\right)+c2,
\end{equation}
where $c_1=\left(\frac{z}{z-2}\right)$ and $c_2= \left(\frac{z}{z-2}\right) \log\left(\frac{1}{\xi}\right)$ 
are both constant with respect to $n$. 
To find the scaling of the upper-bound, we need to calculate the 

\[\log\left(\frac{1}{ l}\right)=\log\left(\frac{1}{ p_0 e^{-(2^n-1)\epsilon}}\right)=\log\left(\frac{(1-(e^{-2\epsilon})^{2^n})}{ (1-e^{-2\epsilon}) e^{-(2^n-1)\epsilon}}\right). \]
For the values of $l$ and $p_0$ see \cite{raeisi2015asymptotic}. 

To understand the scaling, we take $n\gg1$ which simplifies the bound  to 
\[\log\left(\frac{1}{ e^{-(2^n-1)\epsilon}}\right)-c_3,\] 
where $c_3=\log\left(1-e^{-2\epsilon}\right)$. So the scaling of the upper-bound is $O(2^n)$.

\subsection*{The Circuit for $U_{\text{TS}}$}
The unitary that sorts the density matrix lexicographically is
\begin{align}
U _{\text{TS}}=
 \begin{pmatrix}
  1 & 0 & 0 & 0 & 0 &\cdots & 0 & 0 & 0 \\
  0 & 0 & 1 & 0 & 0 &\cdots & 0 & 0 & 0 \\
  0 & 1 & 0 & 0 & 0 &\cdots & 0 & 0 & 0 \\
  0 & 0 & 0 & 0 & 1 &\cdots & 0 & 0 & 0 \\
  0 & 0 & 0 & 1 & 0 &\cdots & 0 & 0 & 0 \\
  \vdots  & \vdots & \vdots   & \vdots & \vdots & \ddots &   \vdots  & \vdots & \vdots  \\
  0 & 0 & 0 & 0 & 0 &\cdots & 0 & 1 & 0 \\
  0 & 0 & 0 & 0 & 0 &\cdots & 1 & 0 & 0 \\
  0 & 0 & 0 & 0 & 0 &\cdots & 0 & 0 & 1
\end{pmatrix}. \label{SM:sorting_unitary}
\end{align}
$U$ is a block diagonal $2^n\times 2^n$ matrix with $1\times1$ blocks in the upper left and lower right corners and the other blocks are
$X=\begin{pmatrix} 0 & 1 \\ 1 & 0\end{pmatrix}$.

The first and the last operations are SHIFT$_m$ operators which are defined as
\begin{equation}
\text{SHIFT}_m\ket{x_1x_2\dots x_n}=\ket{\left(x_1x_2\dots x_n + m\right) mod 2^n}.
\label{SM:Shift}
\end{equation}
This notation is mostly symbolic but should be understood as binary addition between strings labeling states on $n$ qubits.
For example, SHIFT$_1\ket{101}=\ket{110}$.  The other operations are multiple-control-Toffoli, Toff$_n$, which is a controlled NOT with $n-1$
controls and NOT (Pauli X) on the last qubit. 

NOTs are easy to implement and for SHIFT$_1$ and Toff$_n$ 
we use the construction in \cite{saeedi2013linear}.

After Toff$_n$ we apply NOT on the last qubit. Multiple-control-Toffoli and NOT give together the unitary
\begin{align}
U'=
 \begin{pmatrix}
  0 & 1 & 0 & 0 &\cdots & 0 & 0 & 0 & 0\\
  1 & 0 & 0 & 0 &\cdots & 0 & 0 & 0  & 0\\
  0 & 0 & 0 & 1 &\cdots & 0 & 0 & 0  & 0\\
  0 & 0 & 1 & 0 &\cdots & 0 & 0 & 0  & 0\\
  \vdots  & \vdots   & \vdots & \vdots & \ddots &   \vdots & \vdots  & \vdots & \vdots  \\
  0 & 0 & 0 & 0 &\cdots & 0 & 1 & 0 & 0\\
  0 & 0 & 0 & 0 &\cdots & 1 & 0 & 0 & 0\\
  0 & 0 & 0 & 0 &\cdots & 0 & 0 & 1 & 0\\
  0 & 0 & 0 & 0 &\cdots & 0 & 0 & 0 & 1\\
\end{pmatrix}. \label{middle_unitary}
\end{align}
Therefore, to apply U, we must  shift all the rows and columns of $U'$  cyclically. This is what SHIFT$_{+1}$ and its conjugate
transpose SHIFT$_{-1}$ do. We can implement SHIFT$_{+1}$ with Quantum Fourier Transform and rotations \cite{beth2001quantum}. In Fig (\ref{shift}), we show how to implement a
more general operation, SHIFT$_{+m}$.

In this circuit, first, QFT$^{-1}$ transforms the bit-strings from registers to phases
\begin{equation}
\ket{x_{1}}\ket{x_{2}}\dots\ket{x_{n}}\rightarrow \frac{1}{\sqrt{2^{n}}}
\left(\ket{0}+e^{2\pi i\frac{x_{1}x_{2}\dots x_{n}}{2^{n}}}\ket{1}\right)
\left(\ket{0}+e^{2\pi i\frac{x_{2}\dots x_{n}}{2^{n-1}}}\ket{1}\right) \cdots
\left(\ket{0}+e^{2\pi i\frac{x_{n}}{2}}\ket{1}\right).
\end{equation}

In the next step, we apply rotations around $z$ on each qubit. This yields
\begin{equation}
\frac{1}{\sqrt{2^n}}\left(\ket{0} + e^{2\pi i\frac{ x_1x_2\dots x_n + m }{2^n}}\ket{1} \right)\left(\ket{0} + e^{2\pi i\frac{ x_2\dots x_n + m}{2^{n-1}}}\ket{1} \right)\cdots \left(\ket{0} + e^{2\pi i\frac{x_n + m}{2}}\ket{1} \right).
\end{equation}
After applying the Quantum Fourier Transform, we get the desired state $\ket{x_1x_2\dots x_n +m}$.

\begin{figure}[t]
\centering
\includegraphics[width=0.4\columnwidth]{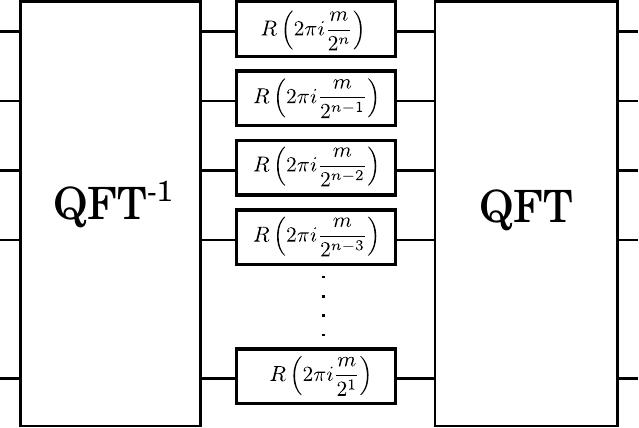}
\caption{The circuit for the SHIFT operator in Eq. (\ref{SM:Shift}).}
\label{shift}
\end{figure}

\subsection*{Complexity of the sort operations in PPA}

First, note that the implementation cost of a sort operation 
 could in general be exponential in $n$. 
Here we give a lower bound for the implementation cost 
that takes into account the cost of the sort operations. 
We use numerical evidence to show that the lower-bound scales 
exponentially with $n$. 

A permutation matrix basically permutes the basis and we need 
to find the circuit with the minimum number of gates that 
gives the same basis transformation. 
Clearly, mapping every element of the basis would be exponential 
in $n$ since there are $2^n$ basis to map. 
So we need to parallelize as many mappings as possible. 

As explained in the text, we assume that all the cycles of the
permutation have efficient implementation and for cycles of
the same size, it is easy to implement them in parallel. 
Without parallelizing  the implementation of blocks that have different sizes,
the complexity would be proportional to the number of blocks
with different size (NBDS). 

The numerical evidence presented in figure (5) in the text, 
shows that NBDS scales exponentially with $n$.

\subsection*{Non-diagonal density matrices}

It is possible that the density matrix becomes non-diagonal during the cooling process due to the noise in the system.
Here we show that our technique
can handle non-diagonal density matrices as well and would take them to the  
OAS.

We show that the presence of any non-diagonal elements will not affect the action of  $U_{\text{TS}}$ on the diagonal  of the density matrix. In other words, the same permutation is applied on the diagonal elements even if the density matrix is non-diagonal in computational basis.

We can write $U_{\text{TS}}$ as
\begin{equation}
U_{\text{TS}} = \ket{0}\bra{0} + \ket{2^{n}-1}\bra{2^{n}-1} + \sum_{i\in odd}^{2^{n}-2}{(\ket{i+1}\bra{i}+\ket{i}\bra{i+1})}.
\end{equation}
Now, if we apply the evolution to some arbitrary density matrix 
$\rho = \sum_{i,j}{\rho_{ij} \ket{i}\bra{j} }$, for the 
diagonal elements we get
\begin{equation}
\bra{i} U_{\text{TS}}\rho U^{\dagger}_{\text{TS}}\ket{i} = \rho_{0,0}\ket{0}\bra{0} + 
\rho_{2^{n}-1,2^{n}-1}\ket{2^{n}-1}\bra{2^{n}-1} + 
\sum_{i\in odd}{(\rho_{i,i}\ket{i+1}\bra{i+1}+\rho_{i+1,i+1}\ket{i}\bra{i})}. 
\end{equation}
The first two terms show that the first and last elements of 
the diagonal elements of the density matrix are preserved, while the 
last term indicates that for all odd values of $i$, we get 
$\rho_{i,i} \leftrightarrow \rho_{i+1,i+1}$. 
This gives the same permutation for the diagonal element as for a diagonal 
density matrix. This means that the diagonal element would again follow 
 a time-independent Markov chain given by the same transfer matrix, T 
(Eq. (5) of the main text). 

Note that polarization only depends on the 
diagonal elements of the density matrix, 
so it suffices to focus on the dynamics 
of these diagonal elements.
As was shown here, this dynamics would be the same regardless of
off-diagonal elements of the density matrix. 

\subsection*{Robustness}

Here we included some additional numerical simulation for the noise
model investigated in figure (3) of the main body, for larger $n$ and polarization, $\epsilon$. 
Specifically, we redid the simulation of PPA for $n=10$ computation and one
reset qubit, and for polarization of $0.1$ and with $\sigma = 0.1, 0.15, 0.2$.

\begin{figure}
\begin{centering}
\includegraphics[trim={0cm 0cm 0cm 0cm},width=.7\columnwidth]{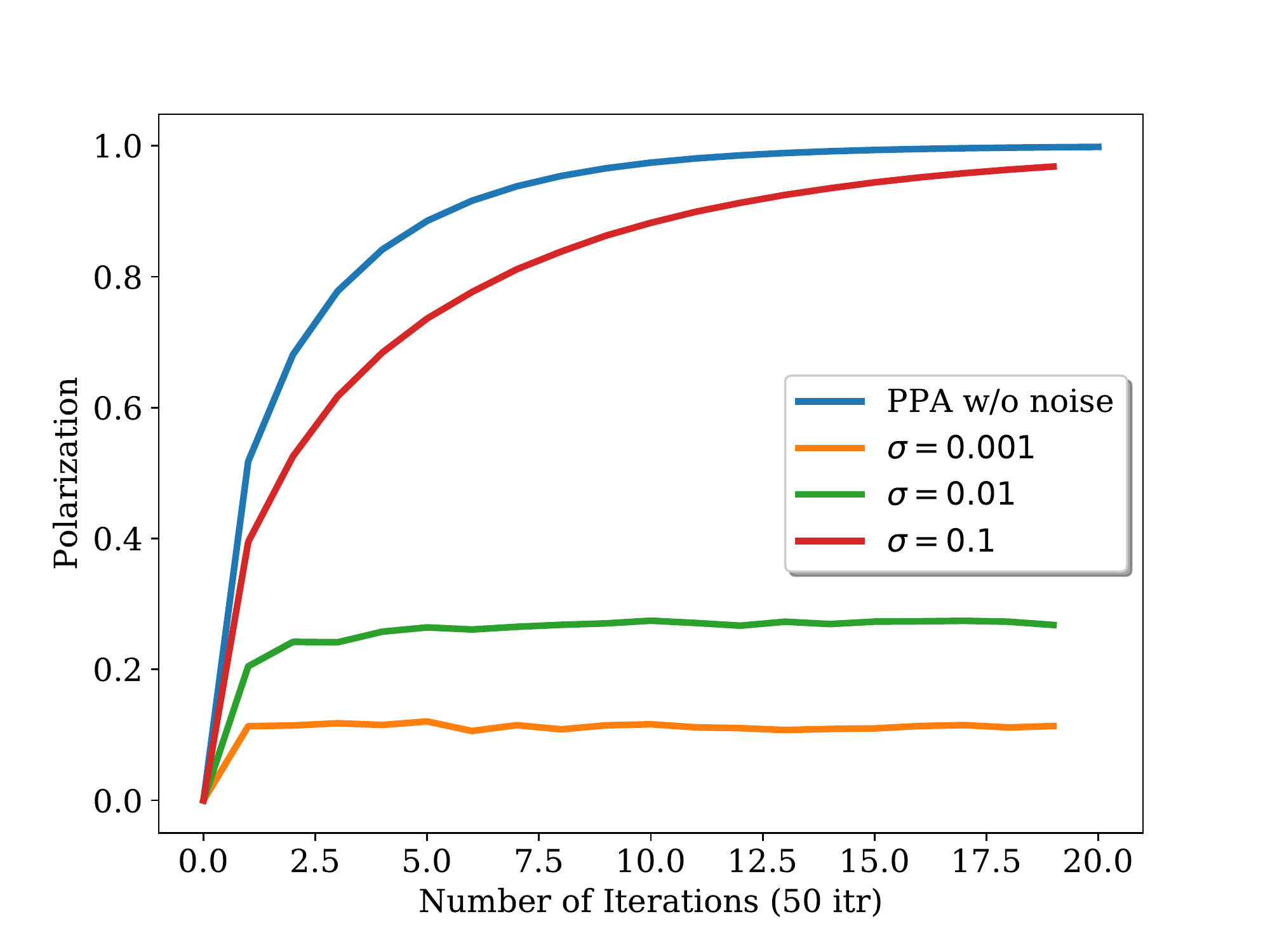}
\par\end{centering}

\caption{\label{fig:sm_Noisy_Tomography} 
This plot shows the polarization of the first qubit vs. iterations. 
This is for the simulation of PPA for $10$ computation and one
reset qubit, polarization of $0.1$ and with different amounts of the state estimation errors, $\sigma$ indicated in the legend of the plot.  }
\end{figure}

\subsection*{Practical Imperfections}
Here we discuss imperfections for practical applications.

There are two mechanisms for error that can be distinguished.

The first is due to the deviation between the state
we assume we have and the state we actually have.
PPA uses information about the state to find
the required sort operation for each iteration. This is in contrast
to TSAC which operates independently of the state.
This is the notion of robustness that was discussed
in the main body.

Second are errors directly introduced by errors in the implementation of the required
operations.
As we mentioned in the main body of this paper, it is not clear how either of these two techniques
would handle such imperfections.

Here we sketch an approach for attempting to analyze robustness in more detail, why we expect TSAC to be more resilient than PPA against certain types of noise, and some of the challenges with analyzing the various possible noise models.

With ideal PPA, the ideal input state $\rho = \rho_0$ undergoes a series of unitary transformations $U_1, U_2, \ldots, U_j, \ldots$, with a reset step in between, so that $\rho_j = U_j \mathcal{R}(\rho_{j-1}) U_j^{\dagger}$ where $\mathcal{R}$ denotes the reset step (where the reset qubits are given time to re-equilibrate and ideally nothing happens to the computation qubits) and the compression operator $U_j$ depends on the state $\rho_{j-1}$.

With ideal TSAC the ideal input state $\rho = \rho_0$ repeatedly undergoes the same unitary transformations $U$ (that does not depend on the input state), with a reset step in between, so $\rho_j = U \mathcal{R}(\rho_{j-1}) U^{\dagger}$.

A noisy PPA may be modelled as instead mapping  (where $\rho_0^{\prime}  =\rho_0$)
\[ \rho_{j-1}^{\prime} \mapsto \rho_j^{\prime} = U_j C_j (\mathcal{R}(\rho_{j-1}^{\prime})) U_j^{\dagger} \] for some completely positive map $C_j$.

Similarly a noisy TSAC may be modelled as instead mapping
\[ \rho_{j-1}^{\prime} \mapsto \rho_j^{\prime} = U C(\mathcal{R}(\rho_{j-1}^{\prime})) U^{\dagger} \] for some completely positive map $C$.

In these characterizations, the $C_j$ (or $C$) operators wlog capture the noise or imperfections while attempting to implement the ideal $U_j$ (or $U$) and we assume also capture the deviation between the ideal reset operation and the actual reset transformation.  
We note that technically not all imperfect implementations can be decomposed in this way. 
However, for the sake of just outlining our intuition here, we will model the noise this way (and we also note that we expect that an actual physical system designed to approximate the HBAC assumptions would have a reset operation that can be closely approximated by the ideal reset operation followed by additional noise on the reset and computation qubits).

The noise operators $C_j$ (or $C$) can affect the cooling process in two ways.
First, they can reduce the purity and directly heat up the state (e.g. a depolarizing channel would do this). Second, the change in the
state may make the subsequent cooling operation ineffective or even lead to cause heating.

More specifically, with PPA, the ideal $U_1$
is designed to sort the diagonal elements
of $\mathcal{R}(\rho_0)$ and thereby cool the state. However, if the diagonal elements of $C_1(\mathcal{R}(\rho_0) )$ have a different order, PPA could end up heating the state.
Similarly, subsequent $U_j$ are designed to cool the ideal input state $\mathcal{R}(\rho_{j-1})$, but may in fact end up heating up the actual state $C(\mathcal{R}(\rho_{j-1}^{\prime}))$ that has been impacted by imperfect operations.

For the first effect, i.e. the direct heating caused by the noise channel,
one could hope that if the heating rate of $C_j$ (or $C$) is less than the cooling
rate of $U_j$ (or $U$), we could show that $U_j C_j (\mathcal{R}( \rho_{j-1}^{\prime})) U_j^{\dagger}$  (or $U C (\mathcal{R}(\rho_{j-1}^{\prime})) U^{\dagger}$) would still improve the purity compared to $\mathcal{R}(\rho_{j-1}^{\prime})$. However, both the heating rate of $C_j$ (or $C$) and cooling rate of $U_j$ (or $U$) depend on their input states, making it challenging to determine if there is a net cooling or heating after both operations are applied.

For the second effect, TSAC seems to have an advantage compared to PPA.  Theorem (1) in the main text implies that, assuming  $C_j (\mathcal{R}(\rho_{j-1}^{\prime}))$ is hotter than
the optimal asymptotic state (OAS), then $U C_j (\mathcal{R}(\rho_{j-1}^{\prime})) U^{\dagger}$ would have a lower temperature than $C_j (\mathcal{R}(\rho_{j-1}^{\prime}))$. But it may not necessarily have a lower temperature than $\mathcal{R}(\rho_{j-1}^{\prime})$.
In contrast, for techniques like PPA,
since the compression is designed for the idea $\rho_{j-1}$, it may not only be
ineffective for $C_j(\mathcal{R}(\rho_{j-1}^{\prime}))$ but it may even heat it up.

In summary, we have outlined some intuition on
the complexity of analyzing imperfections in the implementation of compression operators, and why there may exist noise models under which the
noisy implementation would heat the state,
and other noise models where the state would get cooled.
This topic needs further investigation.

\end{document}